\documentclass[preprint,10pt]{elsarticle}

\makeatletter
\def\ps@pprintTitle{%
\let\@oddhead\@empty
\let\@evenhead\@empty
\def\@oddfoot{\centerline{\thepage}}%
\let\@evenfoot\@oddfoot}
\makeatother

\usepackage{graphicx,bbm,setspace,amssymb,amsfonts,amsmath,mathtools,mathrsfs,stmaryrd,amsthm,bm,etoolbox,comment,hyperref,url,caption,subcaption,color,enumitem,algorithm,booktabs,microtype,nicefrac,lingmacros,nccmath,tree-dvips,tikz-cd}
\usepackage[noend]{algpseudocode}
\usepackage{algorithm}

\journal{Physica A}

\usepackage[colorinlistoftodos]{todonotes}
\usepackage[utf8]{inputenc} 
\usepackage[T1]{fontenc}    

\DeclareMathOperator*{\argmax}{argmax}
\usepackage[noend]{algpseudocode}

\newtheorem{theorem}{Theorem}[section]

\newtheorem{thm-defn}[theorem]{Theorem/Definition}

\newtheorem{prop}[theorem]{Proposition}

\theoremstyle{definition}

\theoremstyle{remark}

\DeclareMathOperator{\E}{\mathbb{E}}

\newcommand{\ignore}[1]{}{}

\begin{document}

\begin{frontmatter}

\title{Global and regional changes in carbon dioxide emissions: 1970-2019}
   
\author[label1]{Nick James} \ead{nick.james@unimelb.edu.au}
\author[label2]{Max Menzies} \ead{max.menzies@alumni.harvard.edu}
\address[label1]{School of Mathematics and Statistics, University of Melbourne, Victoria 3010, Australia}
\address[label2]{Beijing Institute of Mathematical Sciences and Applications, Tsinghua University, Beijing 101408, China}

\begin{abstract}
We introduce new frameworks to study spatio-temporal patterns in carbon dioxide emissions, demographic trends and economic patterns across 50 countries over the past 50 years. Our analysis is broken up into four sections. First, we introduce a new method to classify countries into one of three characteristic emissions classes based on a one, two or three-segment piecewise linear model. We reveal that most countries are best represented by a piecewise linear model with one change point. Next, we perform a decade-by-decade study of carbon dioxide trajectories. There, we demonstrate notable changes in cluster structures in each decade. We then study the spatial propagation of emissions over time, highlighting a peak in spatial dispersion in 2000, beyond which there has been a gradual decline in spatial emissions variance over space. Finally, we use carbon dioxide, GDP, and population data and apply dimensionality reduction and clustering to group countries based on similarity in their real and carbon economies.

\end{abstract}


\end{frontmatter}

\section{Introduction}
\label{Introduction}

Currently, the greatest threat to our planet and the long-term prosperity of human life may be climate change. Climate change refers to long-term structural shifts in weather patterns and temperatures, largely driven by human activities such as the burning of fossil fuels and emission of carbon dioxide (CO$_2$). The long term effects of climate change include ocean warming, rising sea levels, and ice melt. Climatic changes have caused an estimated 150 000 deaths per year since 2002 \cite{WHO_CO2}, and the World Health Organization predicts that this level will rise to 250 000 deaths per year between 2030 and 2050 \cite{WHO_CO3}. Thus, efforts by different governments to reduce CO$_2$ emissions are of the utmost importance for mitigating this ongoing crisis. Mathematical analysis of CO$_2$ emissions over time may identify countries that have been particularly successful in reducing their emissions and provide exemplars for other countries.

Existing research on CO$_2$ emissions is broad, with diverse scientific focuses and methodologies. First, the underlying drivers of such emissions, often related to population dynamics, urbanisation and other societal factors has been covered broadly in the literature \cite{Ribeiro2019,Adams2019,SaintAkadiri2020,Shindell2018,Zhou2019,Dowell2020,NechitaBanda2018,Serrenho2017,Bowerman2011}. In particular, emissions are closely tied to economic growth, so their relationship has been studied closely \cite{Wang2019,Naz2018,Tong2020,Bekun2019}. Numerous approaches have been taken to forecasting emissions \cite{Liu2022,Mason2018,Sutthichaimethee2018,Chiu2020}, including country-specific trends \cite{Tbelmann2020,Peters2019} and offering solutions to the climate change crisis, such as carbon storage \cite{Holloway2007,Zelikova2020}. Some papers have even studied spatio-temporal trends in CO$_2$ emissions \cite{Gao2020,Hu2020,Roobaert2019}. Our paper takes a different approach, focusing on a thorough and descriptive analysis of trends in emissions on a country-by-country basis over the past 50 years. We focus our analysis on segmenting this entire window of analysis into smaller periods, revealing times of the greatest growth in emissions, the times of the largest geographic prevalence of emissions, and relationships with economic and population indices.

Our analysis relies on a broad existing literature of time series analysis as used by applied mathematicians and statisticians. Such methods have been applied widely within the physical and social sciences, including epidemiology  \cite{james2021_CovidIndia,Chowell2016,Manchein2020,james2021_TVO,Blasius2020,James2021_virulence}, finance \cite{arjun,Drod2020_entropy,Jamesfincovid,Drod2021_entropy,james2021_mobility,james2021_crypto2,james2022_stagflation,james_georg}, and other fields \cite{Vazquez2006,Mendes2018,james2021_hydrogen,Shang2020,james2021_olympics,james2022_guns}. There are a variety of time series analysis frameworks that we draw upon, including parametric models \cite{Hethcote2000,Perc2020}, statistical analysis \cite{james2021_spectral}, techniques built upon distance metrics and similarity, \cite{james2021_MJW,Moeckel1997,Szkely2007,Mendes2019,James2020_nsm}, network analysis \cite{Karaivanov2020,Ge2020,Xue2020}, clustering \cite{Machado2020}, and others \cite{Ngonghala2020,Cavataio2021,Nraigh2020,Glass2020}.

This paper proceeds as follows. In Section \ref{sec:Characteristic_classes}, we introduce a new method to identify characteristic emissions classes based on one of three piecewise linear models that best represents their emissions profile over the past 50 years. Next, Section \ref{sec:Trajectories} investigates this evolution more precisely, studying the collective similarity in emissions trajectories on a decade-by-decade basis between 1970 and 2019. We then explore the spatial propagation of emissions over time in Section \ref{sec:Geodesic_variance}, showing that CO$_2$ emissions around the world became more spatially disparate until 2000, followed by increased geographic concentration, primarily driven through China's rapid growth. Finally, Section \ref{sec:Real_carbon_economies} compares different countries' real and carbon economies via dimensionality reduction and clustering techniques. We conclude in Section \ref{sec:Discussion}.


\section{Characteristic emissions classes}
\label{sec:Characteristic_classes}

Throughout this paper, we the $n=50$ countries with the greatest total CO$_2$ emissions as of 2019. We list these alphabetically and index them $i=1,...,n$. We study yearly data across $T=50$ years, indexed $t=1,...,50$. Let $x_i(t) \in \mathbb{R}$ be the multivariate time series of the CO$_2$ emissions of country $i$ in year $t$.

In this section, we introduce a new method to classify countries based on their emissions trends over time. Having noticed that many countries exhibit a piecewise linear trend in their CO$_2$ emissions, we introduce a new framework to determine the most appropriate model for each country. We assume that each country's emission behaviours over time are well represented by one of the three following models:
\begin{enumerate}
    \item Model 0 ($M_0$): a linear model with no change points.
    \item Model 1 ($M_1$): a piecewise linear model with one change point (two piecewise linear components).
    \item Model 2 ($M_2$): a piecewise linear model with two change points (three piecewise linear components).
\end{enumerate}

In our algorithmic procedure, we select the optimal number and placement of up to two change points, such that the average $R^2$ of the various fits (one more than the number of change points) is maximised. Given a preference for a simpler model in the circumstance where a more complex one accounts for a similar level (or marginally more) explanatory variance, we introduce a slight penalty for model complexity. Specifically, suppose we have a single country's CO$_2$ emissions $x(t)$ over a period $t=1,...,T$. We maximise the value of the average $R^2$ across piecewise linear models as follows:
\begin{align}
\label{eq:optimalR2}
\argmax_{0 \leq K \leq 2} \argmax_{1=\xi_{0}<...<\xi_{K+1}=T}  \frac{1}{K} \sum_{i=0}^K R^2_{[\xi_{i},   \xi_{i+1}]} - \alpha K;\\
R^2_{[a,b]}= \max_{\beta_0,\beta_1} 1 - \frac{\sum_{t=a}^b (x(t)-(\beta_0+\beta_1t))^2}{\sum_{t=a}^b (x(t) - \bar{x})^2 }.
\end{align}
In the equation above, $\alpha>0$ is a penalty term, mildly penalising increasing the number of change points $K$ (whereby the number of segments is $K+1$). Having inspected the data, we impose a mild degree of regularisation, setting $\alpha=0.01$. As the value of $R^2$ is commonly reported as a percentage, this penalty imposes just a 1\% penalty on increasing the number of permissible segments.

We remark that the linear models above do not impose the assumption of continuity of the linear fits between segments. While this might seem surprising at first, we believe this is well supported by examining the eventual fits on the data, where clear discontinuities in emissions emerge. This is supported in the discussion of Figure \ref{fig:Optimal_model_change_points} below.

\begin{table*}[ht]
\centering
\begin{tabular}{ |p{2.15cm}|p{1.5cm}|p{3.3cm}|p{1.5cm}|}
 \hline
 Country & Best model & Country & Best model  \\
 \hline
 Algeria & $M_0$ & Morocco & $M_1$ \\
 Argentina & $M_1$ & Netherlands & $M_2$ \\
 Australia & $M_0$ & Nigeria & $M_1$  \\
 Austria & $M_2$ & Oman & $M_1$ \\
 Bangladesh & $M_1$ & Pakistan & $M_1$ \\
 Belgium & $M_2$ & Phillippines & $M_1$ \\
 Brazil & $M_1$  & Poland & $M_2$ \\
 Canada & $M_0$ & Qatar & $M_1$ \\
 Chile & $M_0$ & Romania & $M_1$ \\
 China & $M_2$ & Russia  & $M_2$ \\
 Colombia & $M_1$ & Saudi Arabia & $M_1$ \\
 Czechia & $M_1$ & South Africa & $M_0$ \\
 Egypt & $M_1$ & South Korea & $M_0$\\
 France & $M_2$ & Spain & $M_2$ \\
 Germany & $M_0$ & Taiwan & $M_0$ \\
 India & $M_2$ & Thailand & $M_0$ \\
 Indonesia & $M_1$ & Turkey & $M_0$ \\
 Iran & $M_1$ & Turkmenistan & $M_2$ \\
 Iraq & $M_1$ & Ukraine & $M_2$ \\
 Italy & $M_2$ & United Arab Emirates & $M_1$ \\
 Japan & $M_1$ & United Kingdom & $M_2$ \\
 Kazakhstan & $M_2$ & United States & $M_1$ \\
 Kuwait & $M_0$ &  Uzbekistan &  $M_2$ \\ 
 Malaysia & $M_1$ &  Venezuela &  $M_1$ \\ 
 Mexico & $M_1$ & Vietnam  & $M_2$  \\ 
\hline
\end{tabular}
\caption{Country and optimal piecewise linear models, defined in Section \ref{sec:Characteristic_classes}. The optimal model (and placement of change points) is determined by optimising a penalised averaged $R^2$ as in Eq. (\ref{eq:optimalR2}).}
\label{tab:Country_optimal_models}
\end{table*}

\begin{figure*}
    \centering
    \begin{subfigure}[b]{0.48\textwidth}
        \includegraphics[width=\textwidth]{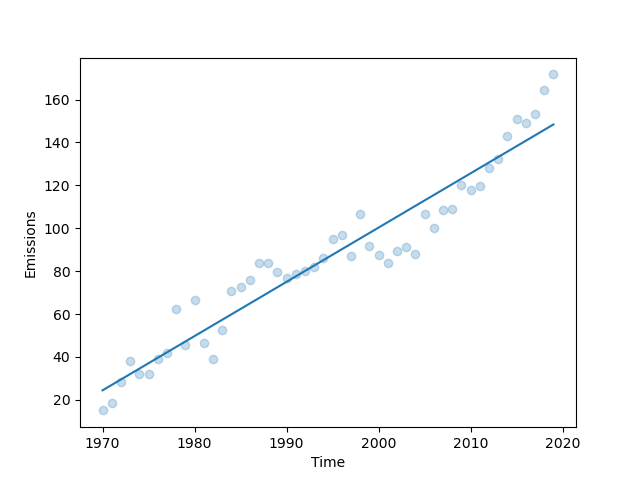}
        \caption{}
    \label{fig:Algeria_optimal_0}
    \end{subfigure}
    \begin{subfigure}[b]{0.48\textwidth}
        \includegraphics[width=\textwidth]{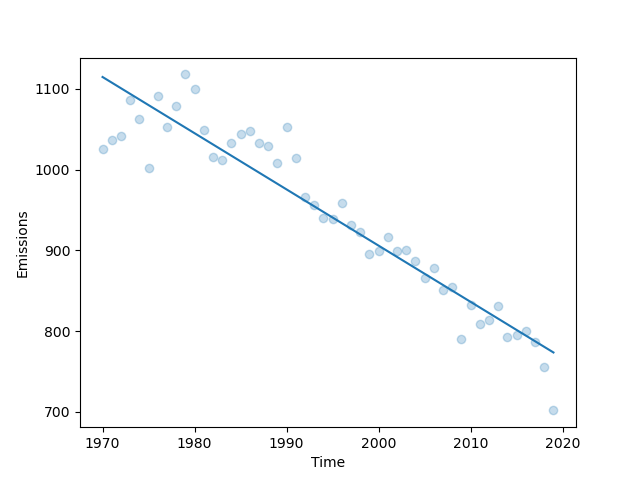}
        \caption{}
    \label{fig:Germany_optimal_0}
    \end{subfigure}
    \begin{subfigure}[b]{0.48\textwidth}
        \includegraphics[width=\textwidth]{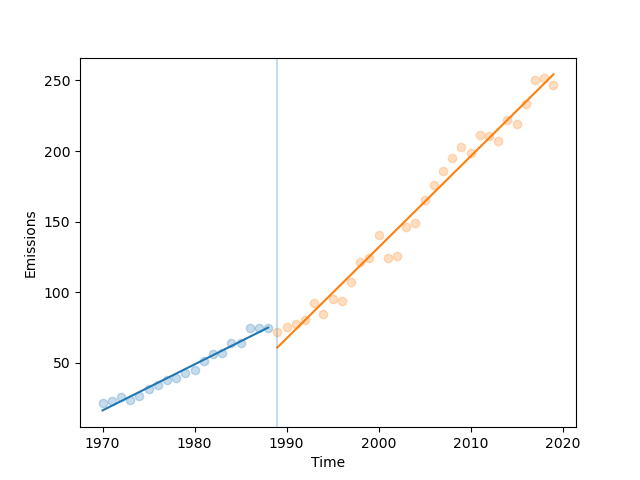}
        \caption{}
    \label{fig:Egypt_optimal_1}
    \end{subfigure}
    \begin{subfigure}[b]{0.48\textwidth}
        \includegraphics[width=\textwidth]{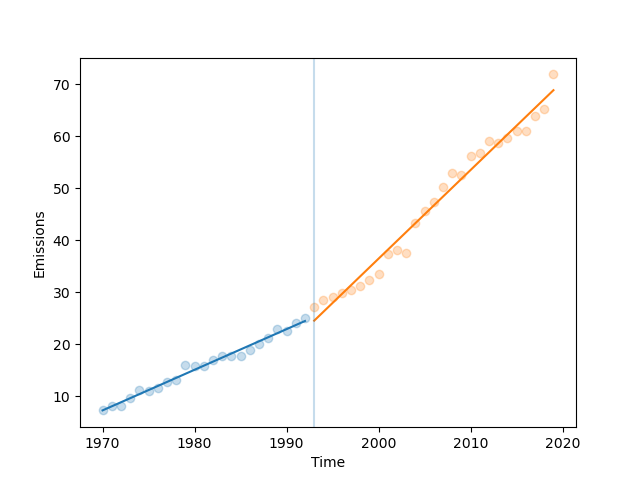}
        \caption{}
    \label{fig:Morocco_optimal_1}
    \end{subfigure}
    \begin{subfigure}[b]{0.48\textwidth}
        \includegraphics[width=\textwidth]{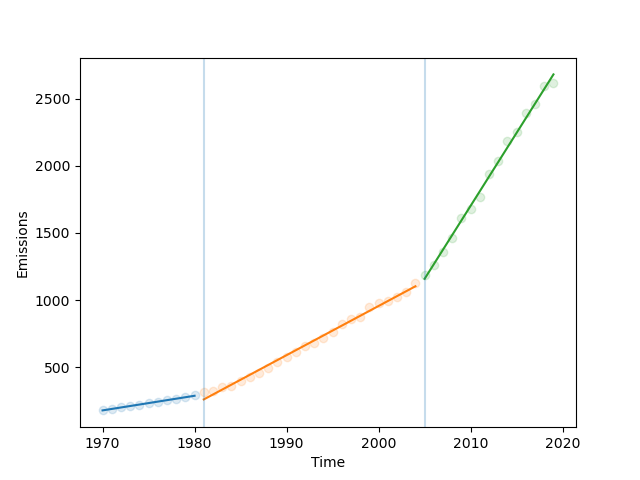}
        \caption{}
    \label{fig:India_optimal_2}
    \end{subfigure}
    \begin{subfigure}[b]{0.48\textwidth}
        \includegraphics[width=\textwidth]{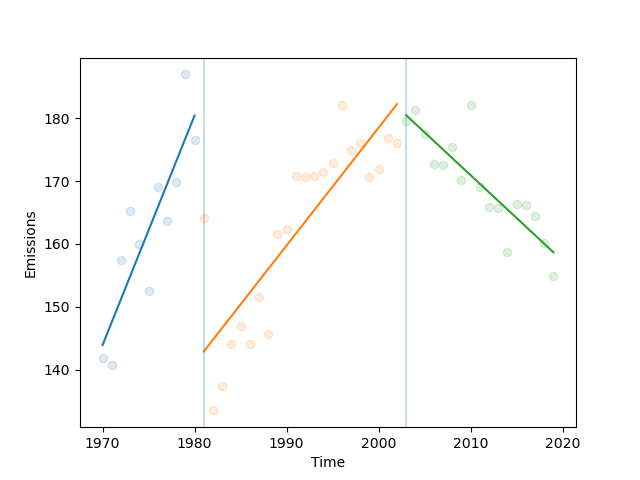}
        \caption{}
    \label{fig:Netherlands_optimal_2}
    \end{subfigure}
     \caption{Country emissions (in billions of metric tonnes) and optimal piecewise linear fits. We provide two examples from each representative model.  (a) Algeria and (b) Germany are best represented by model $M_0$, (c) Egypt and (d) Morocco by model $M_1$, and (e) India and (f) the Netherlands by model $M_2$.}
    \label{fig:Optimal_model_change_points}
\end{figure*}  

We document the optimal model for each country in Table \ref{tab:Country_optimal_models}. There we see the most frequently occurring model is $M_1$ (with 23 countries), followed by $M_2$ (16 countries) and then $M_0$ (11 countries). That is, in the vast majority of cases, CO$_2$ emissions are best modelled by a piecewise linear model with two or three segments, indicating some heterogeneity over time in the growth rate of emissions.

In Figure \ref{fig:Optimal_model_change_points}, we display two representative countries for each model. In Figures \ref{fig:Algeria_optimal_0} and \ref{fig:Germany_optimal_0}, respectively, we display trends for Algeria and Germany, revealing one persistent linear trend over the past 50 years. Algeria exhibits consistently increasing emissions, while Germany exhibits linearly decreasing emissions over the past 50 years. Egypt (Figure \ref{fig:Egypt_optimal_1}) and Morocco (Figure \ref{fig:Morocco_optimal_1}) both exhibit a more strongly positive linear increase beyond the late 1980s and early 1990s respectively. Fittingly, these countries are best modelled by one change point. Both India and the Netherlands, Figures \ref{fig:India_optimal_2} and \ref{fig:Netherlands_optimal_2}, respectively, require two change points (three segments) to best model their dynamics, fitting the numerous changes in their emissions' evolution over time. India exhibits clear continuity and increases in its emissions over time, which are appropriately captured by the adjoined linear fits. On the other hand, the Netherlands features clear discontinuities and abrupt changes in total emissions, seen in the early 1980s and early 2000s, respectively.

\section{Emissions temporal partition analysis}
\label{sec:Trajectories}

\begin{figure*}
    \centering
    \begin{subfigure}[b]{\textwidth}
        \includegraphics[width=\textwidth]{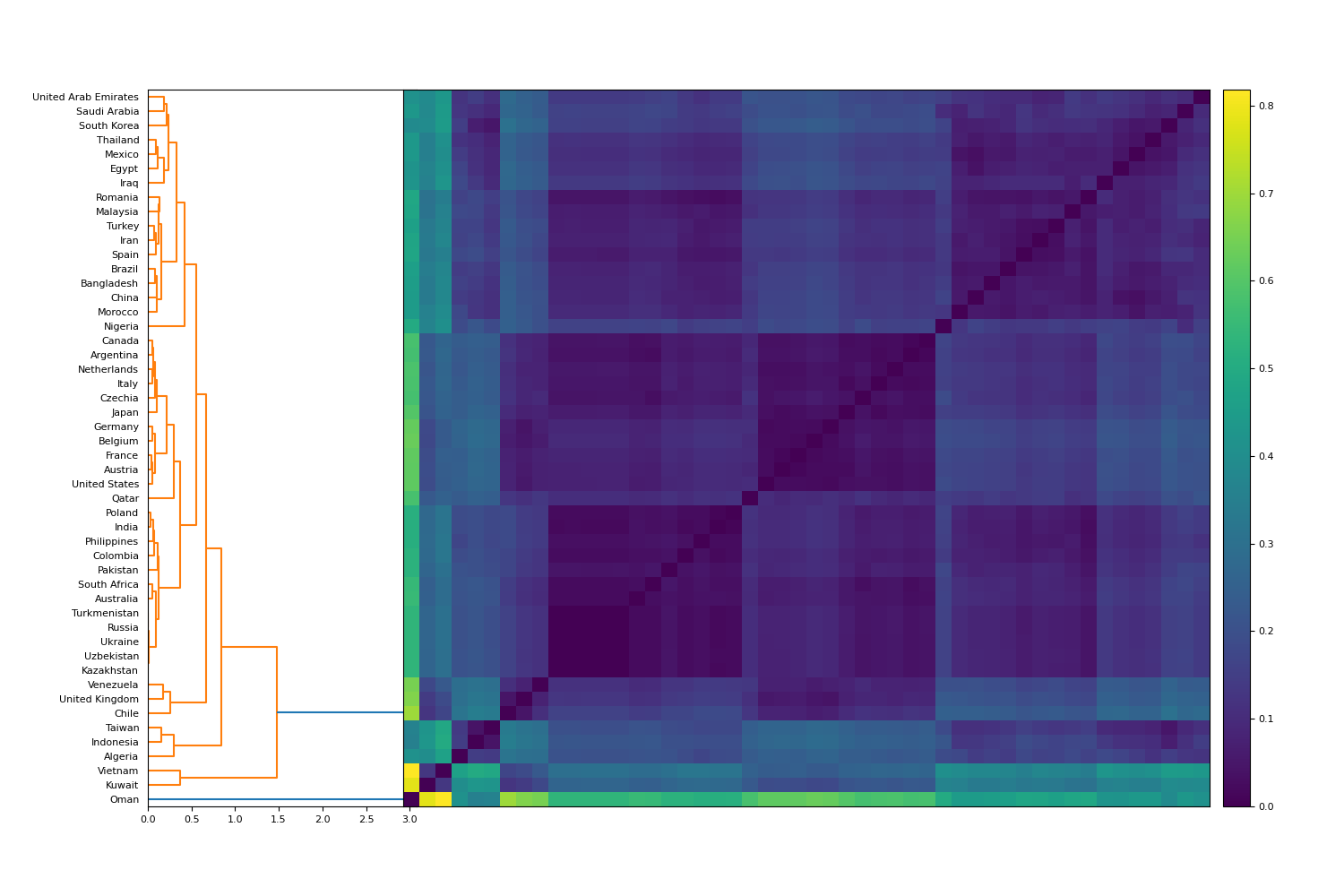}
        \caption{}
    \label{fig:Trajectory_1970_1979}
    \end{subfigure}
    \begin{subfigure}[b]{\textwidth}
        \includegraphics[width=\textwidth]{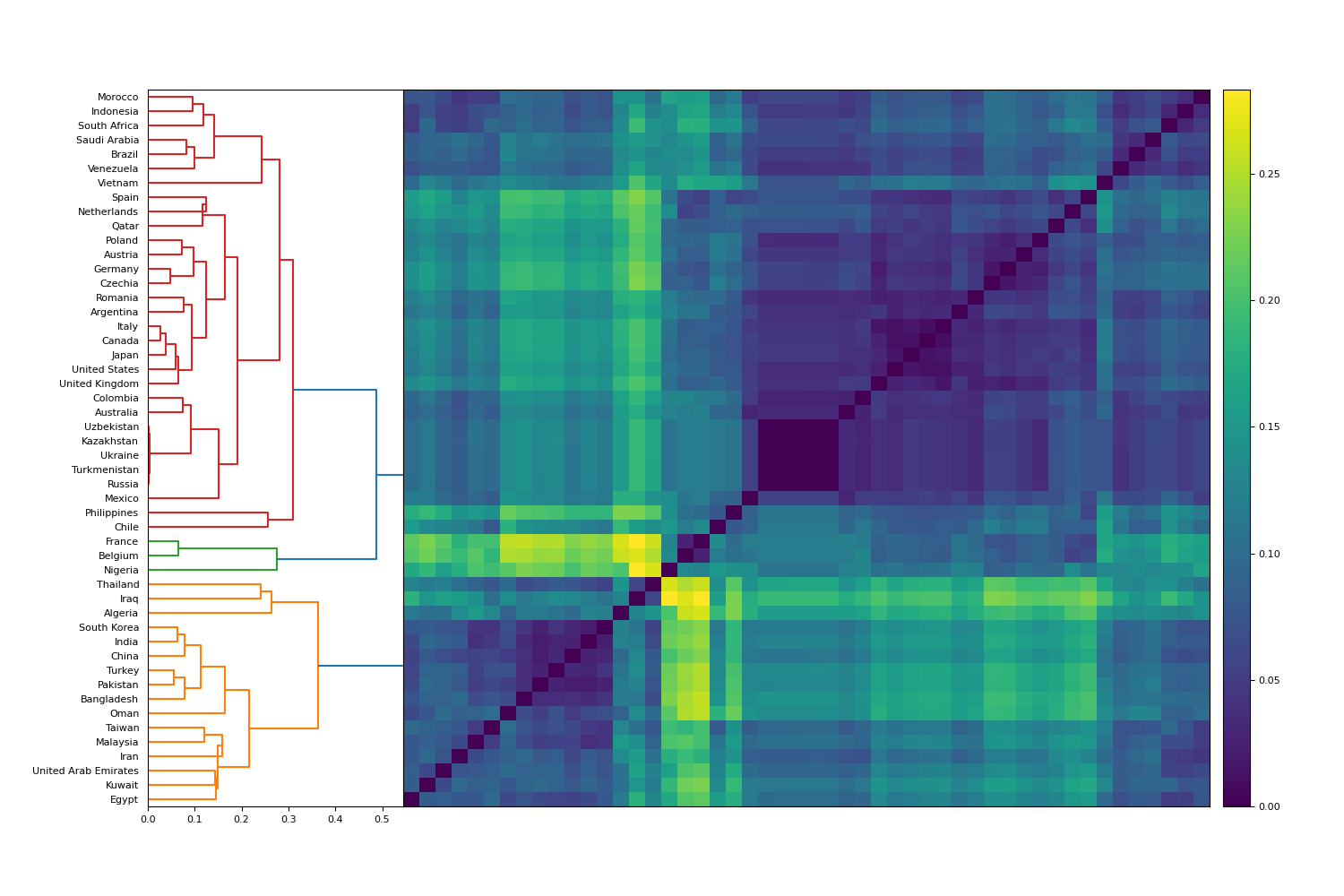}
        \caption{}
    \label{fig:Trajectory_1980_1989}
    \end{subfigure}
\end{figure*}
\begin{figure*}\ContinuedFloat
    \begin{subfigure}[b]{\textwidth}
        \includegraphics[width=\textwidth]{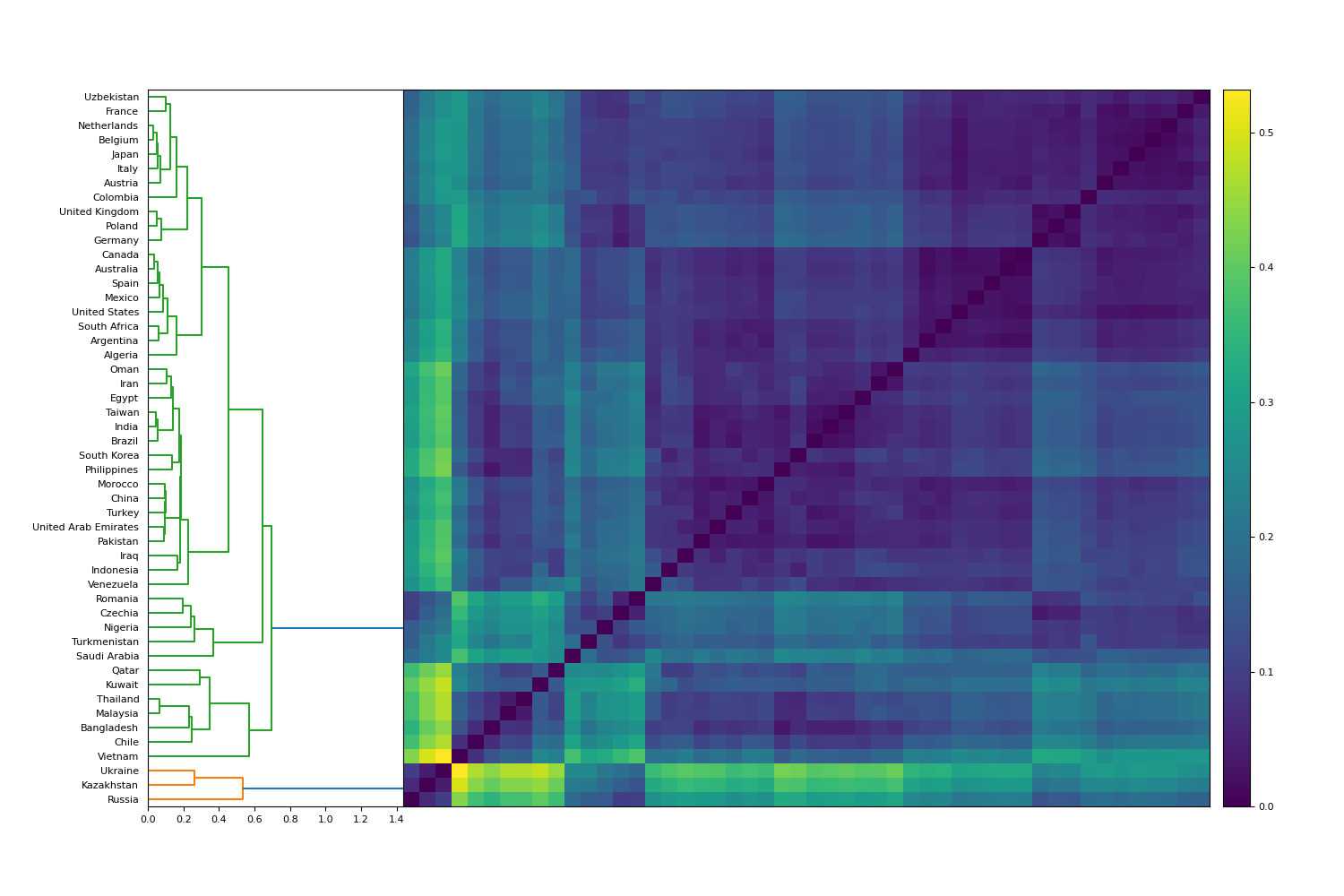}
        \caption{}
    \label{fig:Trajectory_1990_1999}
    \end{subfigure}
    \begin{subfigure}[b]{\textwidth}
        \includegraphics[width=\textwidth]{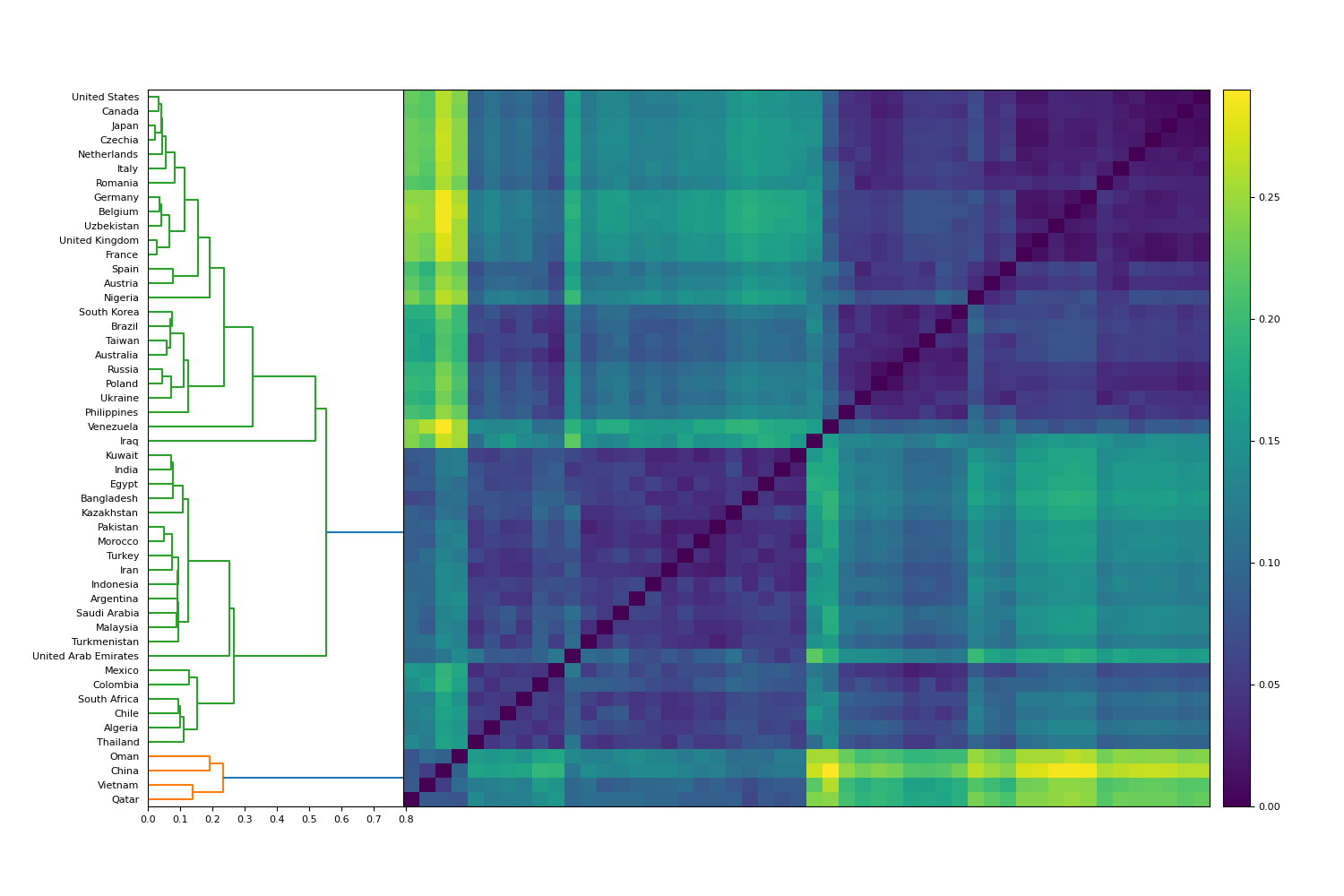}
        \caption{}
    \label{fig:Trajectory_2000_2009}
    \end{subfigure}
\end{figure*}
\begin{figure*}\ContinuedFloat
    \begin{subfigure}[b]{\textwidth}
        \includegraphics[width=\textwidth]{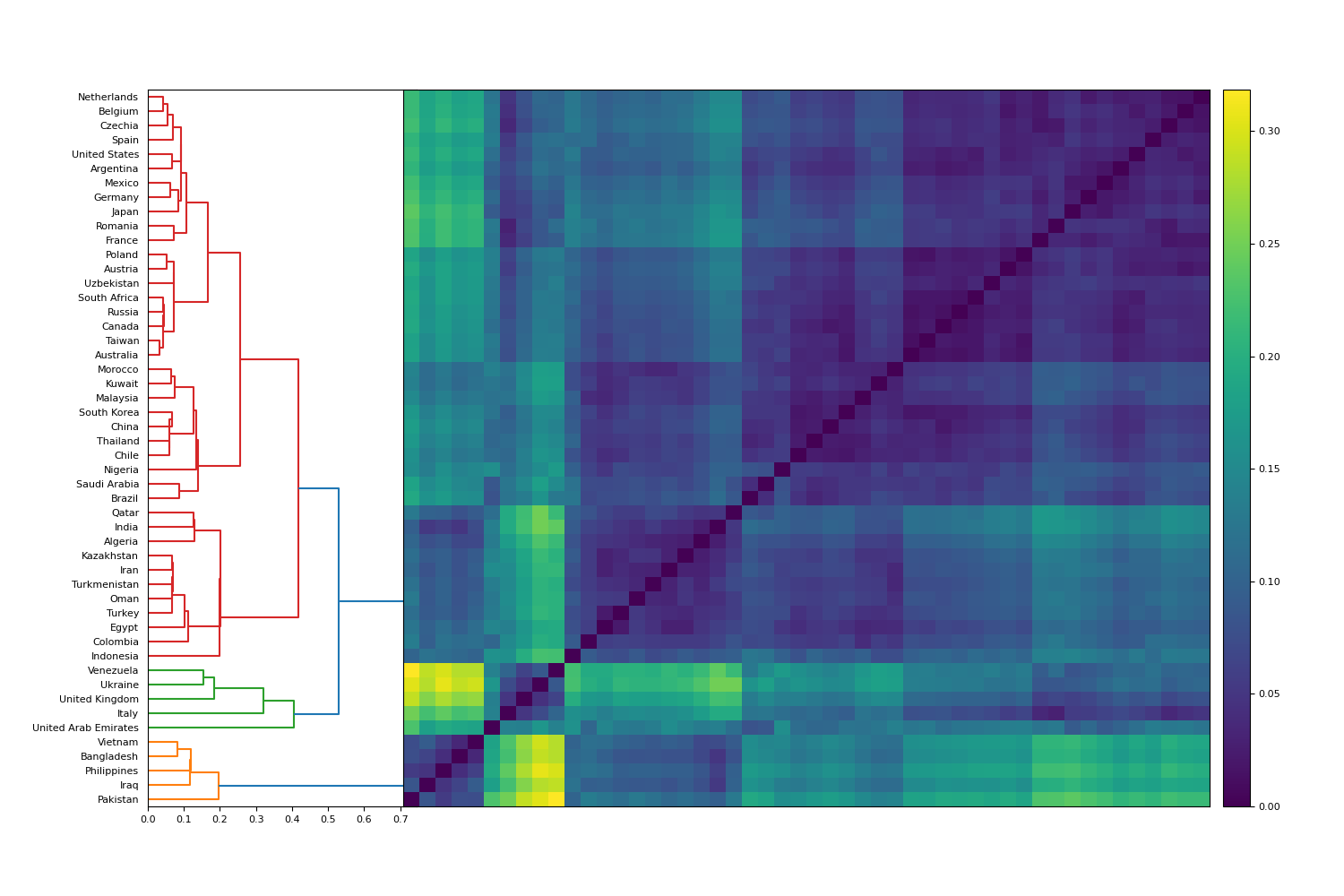}
        \caption{}
    \label{fig:Trajectory_2010_2019}
    \end{subfigure}
    \caption{Hierarchical clustering applied to normalised CO$_2$ emissions trajectories on a decade-by-decade basis. (a) 1970-1979, (b) 1980-1989, (c) 1990-1999, (d) 2000-2009, (e) 2010-2019.}
    \label{fig:Trajectory_clustering_decade}
\end{figure*}  

In this section, we partition country emissions by decade and investigate on a decade-by-decade basis, exploring the evolutionary structure of country behaviours. We restrict each country's emission values to a particular decade, giving a sequence $\mathbf{g}_i=({g}_i(1), {g}_i(2), \dots, {g}_i(m)) \in \mathbb{R}^m$, where $m=10$ is the number of years in that decade, and $i=1,\dots, n$ represent country index. Let $||\mathbf{g}_i||=\sum_{t=1}^m |{g}_i(t)|$ be the $L^1$ norm of the vector $\mathbf{g}_i$. As all ${g}_i(t)$ are non-negative, this counts the total emissions observed in a decade. We may then define $\mathbf{h}_i= \frac{\mathbf{g}_i}{||\mathbf{g}_i||}$. The vectors $\mathbf{h}_i$ reflect relative changes in emissions within a decade.  For example, a country whose emissions in a decade differ between 1000 and 1100 units (of mass) will present a relatively flat normalised trajectory, whereas a country whose emissions in a decade rise from 0 to 100 units will present a more steeply increasing normalised trajectory, as a reflection of the relative change. We define \emph{trajectory distance matrices} $D_{ij}=||\mathbf{h}_i-\mathbf{h}_j||$ that measure distance between these normalised trajectories.

In Figures \ref{fig:Trajectory_1970_1979}, \ref{fig:Trajectory_1980_1989}, \ref{fig:Trajectory_1990_1999}, \ref{fig:Trajectory_2000_2009} and \ref{fig:Trajectory_2010_2019}, we implement hierarchical clustering on these matrices for the ten-year periods of 1970-1979, 1980-1989, up to 2010-2019. Both the number of clusters and cluster constituents are dynamic, with both varying as we proceed forward in time.

First, during 1970-1979, we observe a predominant cluster, consisting of three subclusters, and a small collection of outlier countries (primarily displayed as separate subclusters within the predominant cluster). All subclusters within the predominant cluster exhibit increasing trends of CO$_2$ emissions over the decade. The first subcluster consists of countries such as Russia (Figure \ref{fig:Russia_emissions}) and Ukraine - these countries are characterised by huge growth and accelerating emissions trends. The second subcluster, consisting of Italy, Canada and Argentina, exhibits moderate growth in CO$_2$ emissions. The final cluster, which includes Spain and Brazil, also displays relatively steady growth behaviours. The outlier countries, such as Vietnam and Kuwait, exhibit declining CO$_2$ emissions, which are anomalous with respect to the rest of the collection.

Next, we turn to the period 1980-1989. This period produces a dendrogram consisting of three distinct clusters. The first cluster contains countries such as India (\ref{fig:India_emissions}) and China (\ref{fig:China_emissions}), which display continued growth throughout the period. The second cluster consists of various countries, but displays pronounced similarity between Eastern Europe and Central Asia, countries such as Kazakhstan, Russia and Ukraine. These countries all experienced huge growth in emissions, peaking around 1990. The final cluster consists of France, Belgium and Nigeria. These countries display erratic emissions behaviours, with all trajectories displaying limited trend and substantial volatility.

In the 1990-1999 period, countries form one primary cluster, with a small collection of outlier countries in a separate, significantly smaller cluster. The primary cluster consists of countries that displayed consistent growth over the prior decade (to varying extents). The outlier cluster consists of countries such as Ukraine, Russia and Kazakhstan - these all experienced precipitous drops in their emissions at the beginning of the period, and continuing decline throughout the decade.

The 2000-2009 period produces one primary trajectory cluster (consisting of two similarly sized subclusters) and a collection of outlier countries. The clear bifurcation in the large cluster is indicative of contrasting trends in emissions behaviours between various countries. The first subcluster consists of countries such as the Netherlands, Italy, Germany, the United States (Figure \ref{fig:USA_emissions}), Canada, Japan (\ref{fig:Japan_emissions}) and Belgium. Most countries within this subcluster are more developed and have taken a stronger stance in introducing policies to reduce emissions. Such countries exhibit either flat or declining emissions trajectories throughout the decade. The second subcluster, consisting of countries such as India (\ref{fig:India_emissions}), Iran and Turkey, features sustained growth in emissions during this period. The significantly smaller second cluster consists of countries whose emissions profiles more resemble those of the second subcluster. This cluster consists of China (\ref{fig:China_emissions}), Oman, Vietnam and Qatar.

Finally, we turn to the most recent decade of analysis, 2010-2019. This period produces three characteristic classes of emissions trajectories. The first cluster consists of countries with lower human development index (HDI) levels: Vietnam, Bangladesh, the Philippines, Iraq and Pakistan. These countries produce emissions trajectories that increase significantly throughout the period. The second cluster consists of Venezuela, Ukraine, the United Kingdom (UK), Italy and the United Arab Emirates (UAE). These countries mostly produce declining trajectories, which may signify their greater collective focus on reducing CO$_2$ emissions at the national level. The final cluster consists of two primary subclusters. The first contains countries such as Japan (\ref{fig:Japan_emissions}), the Netherlands and Belgium. These countries are primarily characterised by erratic emissions output, with a declining trend overall. The second subcluster consists of countries such as China (\ref{fig:China_emissions}), Thailand and Chile. Although these countries produce a positive trend over the entire decade, the rate of increase has slowed, and in some cases, overall emissions have begun to trend downward.

\begin{figure*}
    \centering
    \begin{subfigure}[b]{0.48\textwidth}
        \includegraphics[width=\textwidth]{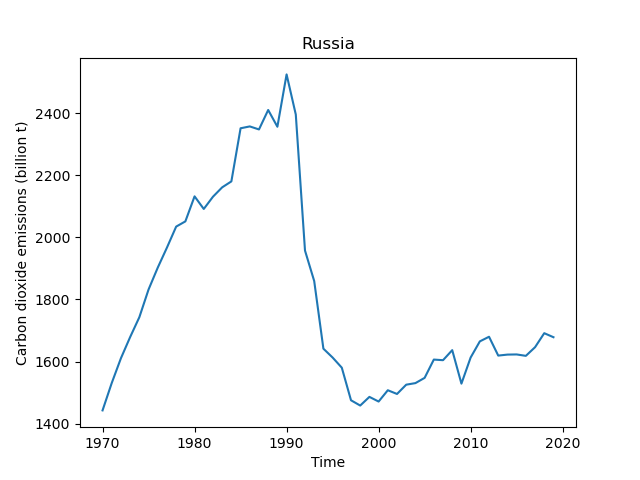}
        \caption{}
    \label{fig:Russia_emissions}
    \end{subfigure}
    \begin{subfigure}[b]{0.48\textwidth}
        \includegraphics[width=\textwidth]{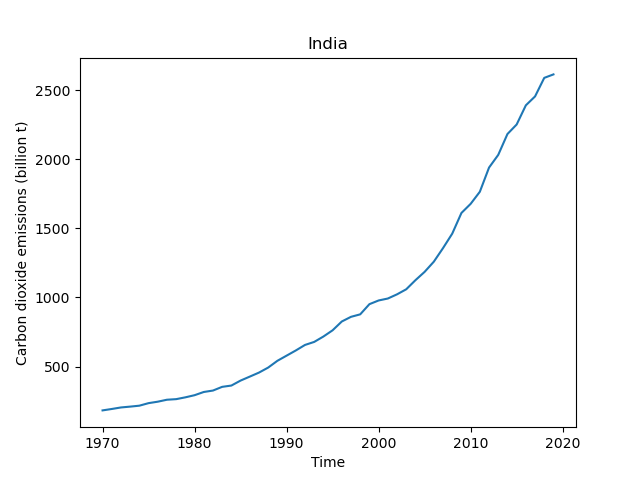}
        \caption{}
    \label{fig:India_emissions}
    \end{subfigure}
    \begin{subfigure}[b]{0.48\textwidth}
        \includegraphics[width=\textwidth]{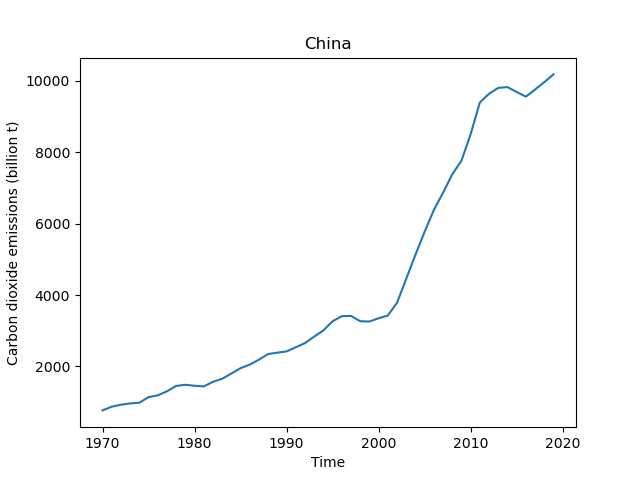}
        \caption{}
    \label{fig:China_emissions}
    \end{subfigure}
        \begin{subfigure}[b]{0.48\textwidth}
        \includegraphics[width=\textwidth]{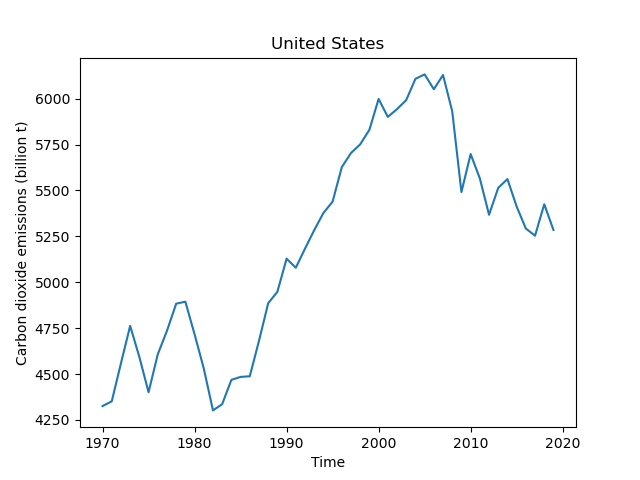}
        \caption{}
    \label{fig:USA_emissions}
    \end{subfigure}
    \begin{subfigure}[b]{0.48\textwidth}
        \includegraphics[width=\textwidth]{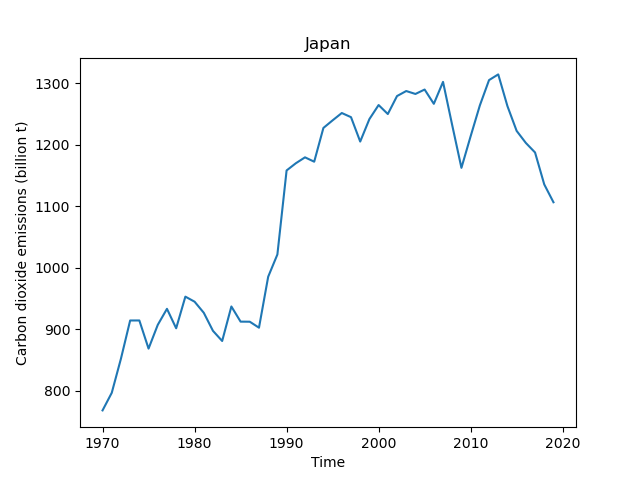}
        \caption{}
    \label{fig:Japan_emissions}
    \end{subfigure}
     \caption{Country emissions (in billions of metric tonnes) over time for (a) Russia (b) India (c) China (d) the United States (e) Japan.}
    \label{fig:Country_emissions}
\end{figure*}


\section{Evolutionary spatial emissions variance}
\label{sec:Geodesic_variance}

In this section, we explore the evolution in geographic dispersion among our 50 countries' CO$_2$ emissions throughout our window of analysis. We study the multivariate time series $x_i(t)$ as a whole with a focus on the changing geographic spread of emissions over time. To do so, we convert each year's emissions values into a probability density function and incorporate real-world distances between countries to produce a measure of the physical spread of emissions across the world. Let $f(t) \in \mathbb{R}^n$ be the \emph{probability vector} of CO$_2$ emissions produced by each country annually, divided by the total amount of emissions produced by our 50 countries that year. Specifically, let 
\begin{align}
f_i(t) = \dfrac{x_i(t)}{\sum_{j=1}^{n} x_j(t)}, i={1},...,n. 
\end{align}
This produces a \emph{probability vector} $f(t) \in \mathbb{R}^n$ for each year $t=1,...,T.$ By a probability vector $v \in \mathbb{R}^n$, we mean a vector $v=(v_1,...,v_n)$ with the property that $0 \leq v_i \leq 1 $ for all $i$ and $\sum_i v_i = 1$. A probability vector could also be called a proportion or percentage vector: in this case, each element of $f(t) \in \mathbb{R}^n$ is the proportion (or percentage or fraction) of the world's emissions in that year by each country. Naturally, the proportions of each country must add up to 1 (or 100\%). The term \emph{probability vector} does have a probabilistic interpretation: $f_i(t)$ is the probability that a random ton of emitted CO$_2$ came from a country $i$. However, the main reason we use the term probability vector is that we will proceed to understand each $f(t)$ as a probability distribution (or probability density function). Then, we may borrow from the rich theory of probability distributions and core concepts like the variance of a probability distribution. 

Specifically, we wish to use these probability distributions together with real-world distances to produce a time-varying measure of global dispersion of emissions. For this purpose, we apply the concept of \emph{geodesic variance} introduced in \cite{James2021_geodesicWasserstein} - this determines the total spread of a probability density function across any candidate metric space. Specifically, given a probability distribution $f$ that corresponds to a measure $\mu$ sitting on a metric space $(X,d)$, we let
\begin{align}
\label{eq:geodesicvariance}
\text{Var}(f) &= \int_{X \times X} d(x,y)^2 d\mu(x) d\mu(y) \\ &= \sum_{x, y \in X} d(x,y)^2 f(x)f(y).    
\end{align}
The second equality is valid when the metric space $(X,d)$ is discrete, as is the case in our analysis. Henceforth, we apply this quantity in the instance where $X$ is the set of 50 countries under consideration and $d(i,j)$ is the real-world haversine \cite{haversine} distance between (the centroids of) any two countries indexed by $i$ and $j$. We refer to (\ref{eq:geodesicvariance}) as the \emph{geographic variance} of the probability distribution $f$; it is computed as follows:
\begin{align}
\label{eq:geodesicvariance2}
\text{Var}(f)  = \sum_{i,j=1}^n d(i,j)^2 f_if_j.    
\end{align}
In \ref{app:proposition}, we explain why we refer to this quantity with the word ``variance''; we include a proof that it generalises the classical notion of the variance of a random variable on the real numbers $\mathbb{R}$.

In Figure \ref{fig:Geodesic_variance}, we plot the geographic variance $\text{Var}(f(t))$ of CO$_2$ emission distributions over time. The plot reveals several interesting findings. First, the trajectory features a broadly quadratic shape with a steady increase in emissions' geographic variance until the year 2000, after which there is a steady decrease in variance. This shape is consistent with the theme that until the year 2000, globalisation and accompanying economic growth led to the propagation of CO$_2$ emissions all around the world, increasing geographic dispersion. The fall in this variance since 2000 indicates that the most serious emissions violators are becoming more spatially concentrated. We further investigate this evolution by studying the proportion of emissions by the three largest countries (by population and emissions): China, India and the United States (US). In 1970, the US was a clear leader in emissions, accounting for 32.7\% of the distribution (across our 50 countries). By contrast, India and China only accounted for 7.8\% and 5.8\% of the distribution, respectively. In the year 2000, the US' share of the total emissions distribution fell to 26.4\%, with China's contribution to the distribution increasing to 14.8\% and India's contribution falling to 4.3\%. In the following 19 years, country emissions' behaviours changed considerably. By 2019, China, the US and India account for 30.7\%, 15.9\% and 7.9\% of the top 50 countries' total emissions, respectively. These figures provide some insight into the overall shape displayed in Figure \ref{fig:Geodesic_variance}. In particular, the reduction in geographic variance exhibited since 2000 could be symptomatic of rapid growth in CO$_2$ emissions produced by China, causing a reallocation of the country's contribution to the overall emissions total, consistent with a larger proportion of emissions coming from a more spatially concentrated area.

\begin{figure*}
    \centering
    \includegraphics[width=0.95\textwidth]{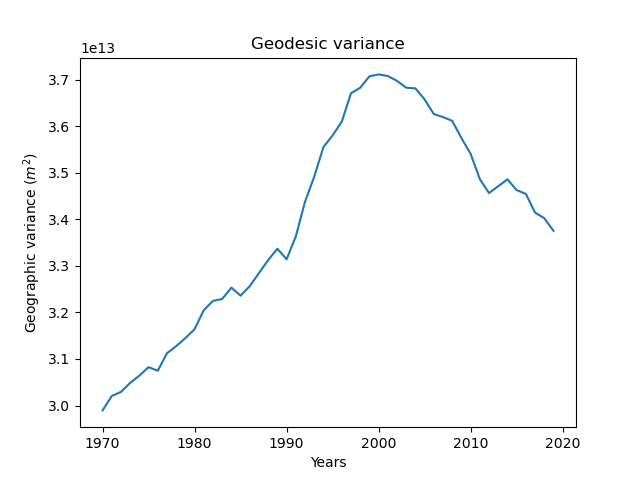}
    \caption{Geographic variance (in $m^2$) in carbon dioxide emissions between 1970-2019. The trajectory displays a noteworthy quadratic shape. Geographic variance in CO$_2$ emissions increase until 2000, beyond which they steadily decrease. This suggests that since 2000 there is less spatial dispersion in the bulk of CO$_2$ emissions. }
    \label{fig:Geodesic_variance}
\end{figure*}

\section{Relationships between the real and carbon economies}
\label{sec:Real_carbon_economies}

In this section, we investigate countries that share similar economic, demographic and emissions histories. First, we gather GDP, population and CO$_2$ emissions data over the past 50 years and generate distance matrices between all countries for each respective metric. Next, we generate an aggregate distance matrix, appropriately normalising for each quantity therein. We then apply multidimensional scaling (MDS) to the distance matrix, projecting the matrix into a lower-dimensional vector space, and then apply K-means clustering to the resulting projection. The number of clusters $k$ is determined by optimising the silhouette score \cite{Rousseeuw1987}. For the sake of interpretability, we accompany these K-means clustering results with a hierarchical clustering dendrogram.

Specifically, we form our normalised aggregate distance matrix by considering the multivariate time series of gross domestic product (GDP) and population data $y_i(t)$ and $z_i(t)$, respectively, in addition to the aforementioned emissions data $x_i(t)$. We define $n \times n$ distance matrices as follows:
\begin{align}
  \label{eq:Xdistance}  D^X_{ij}=\sum_{t=1}^T |x_i(t) - x_j(t) |; \\
 \label{eq:Ydistance}   D^Y_{ij}=\sum_{t=1}^T |y_i(t) - y_j(t) |;\\
    \label{eq:Zdistance} D^Z_{ij}=\sum_{t=1}^T |z_i(t) - z_j(t) |.
\end{align}
For each matrix $D$, let $\| D \|_\infty$ be the maximal element of $D$. Our normalised aggregated matrix is then formed as follows:
\begin{align}
    \Omega = \frac{1}{\| D^X \|_\infty} D^X + \frac{1}{\| D^Y \|_\infty} D^Y  + \frac{1}{\| D^Z \|_\infty} D^Z.
\end{align}
We then apply multidimensional scaling to this matrix to find point projections in 3D space. This procedure selects $z_1,...,z_n \in \mathbb{R}^3$ to minimise the following discrepancy score:
\begin{align}
    \sum_{i,j=1}^n (\Omega_{ij} - \| z_i - z_j\|)^2.
\end{align}
This is a necessary preprocessing step before K-means clustering can be applied. We choose three-dimensional projections as the matrix $\Omega$ is approximately a distance matrix between three-dimensional vectors. Indeed, suppose for two countries $i,j$ we have uniformly for all time, $x_i(t)>x_j(t), y_i(t)>y_j(t)$ and $z_i(t)>z_j(t)$. This is typically seen when a country indexed $i$ is consistently larger than one indexed $j$. Then $\Omega_{ij}$ simplifies as follows. First, let
\begin{align}
\omega_i=\frac{1}{\| D^X \|_\infty} \sum_{t=1}^T x_i(t) + \frac{1}{\| D^Y \|_\infty} \sum_{t=1}^T y_i(t)  + \frac{1}{\| D^Z \|_\infty} \sum_{t=1}^T z_i(t).
\end{align}
Then, under the assumption that $x_i(t)>x_j(t), y_i(t)>y_j(t)$ and $z_i(t)>z_j(t)$ for the country pair $(i,j)$, we have $\omega_i>\omega_j$, and $\Omega_{ij}$ can be reexpressed simply as $\Omega_{ij}=\omega_i - \omega_j$. Perhaps more symmetrically, if the pair of countries $(i,j)$ satisfies the property that one (either $i$ or $j$) is larger in all three attributes than the other, then $\Omega_{ij}$ can be symmetrically rewritten as $|\omega_i - \omega_j|.$ Thus, for many pairs of countries $i,j$, the distance $\Omega_{ij}$ can be expressed as a $L^1$ distance between three scalars, which approximates a Euclidean distance in $\mathbb{R}^3$.

When we apply this analysis to all 50 countries, the US exhibits highly anomalous behaviour. In fact, we must sequentially remove the US, China, India, Russia and Japan (some of the largest countries with respect to emissions) to identify the general structures without these outlier countries. In the proceeding analysis, we remove these five countries, and study the structural patterns among the remaining countries.

\begin{figure*}
    \centering
    \begin{subfigure}[b]{\textwidth}
        \includegraphics[width=\textwidth]{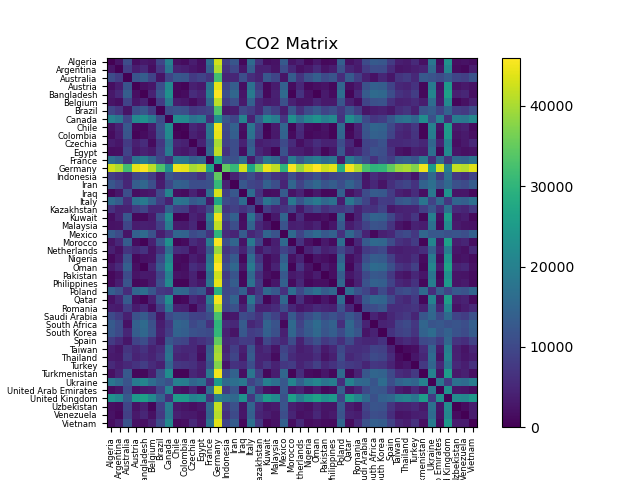}
        \caption{}
    \label{fig:Carbon_dioxide_matrix}
    \end{subfigure}
    \begin{subfigure}[b]{\textwidth}
        \includegraphics[width=\textwidth]{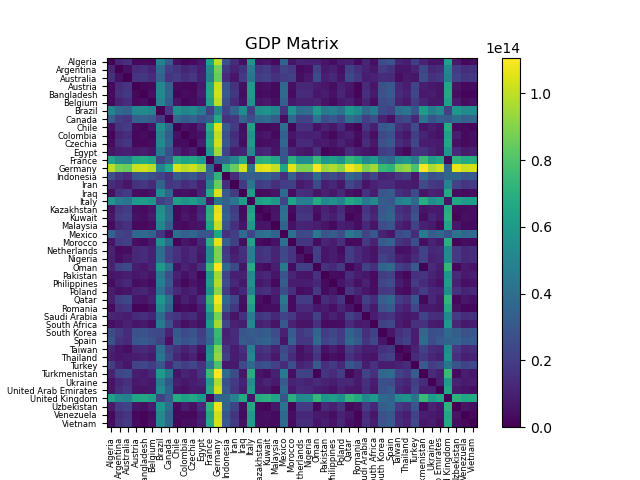}
        \caption{}
    \label{fig:GDP_matrix}
    \end{subfigure}
\end{figure*}
\begin{figure*}\ContinuedFloat
    \begin{subfigure}[b]{\textwidth}
        \includegraphics[width=\textwidth]{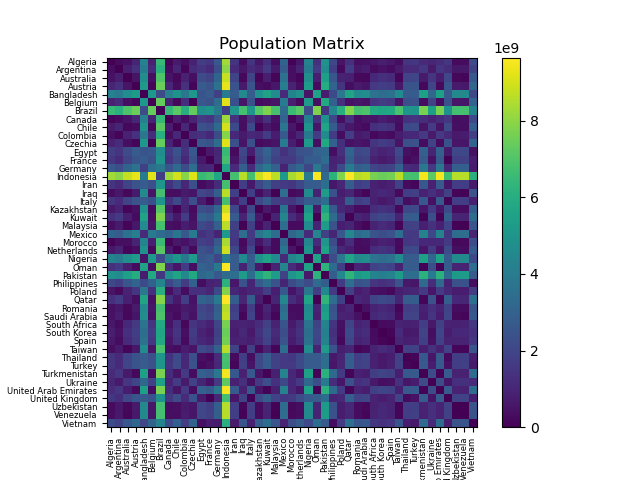}
        \caption{}
    \label{fig:Population_matrix}
    \end{subfigure}
    \caption{Distance matrices between (a) CO$_2$ emissions, defined in (\ref{eq:Xdistance}), (b) GDP, defined in (\ref{eq:Ydistance}), (c) population, defined in (\ref{eq:Zdistance}).}
    \label{fig:Distance_matrices}
\end{figure*}  

First, we explore the similarity between countries with respect to each metric individually. Figure \ref{fig:Distance_matrices} displays the three distance matrices for CO$_2$ emissions, GDP, and population. In Figure \ref{fig:Carbon_dioxide_matrix}, Germany exhibits the greatest dissimilarity to the remainder of the collection of countries. Indeed, as seen in Figure \ref{fig:Germany_optimal_0}, Germany exhibits consistently decreasing emissions over time, unlike almost any other country. In Figure \ref{fig:GDP_matrix}, one can see that Germany, France and the UK are the most anomalous countries with respect to GDP based on their significant economic output over the past 50 years. In Figure \ref{fig:Population_matrix}, we see the distance between country populations, where Indonesia, Brazil and Bangladesh exhibit the greatest dissimilarity to the rest of the collection. 

\begin{figure*}
    \centering
    \includegraphics[width=\textwidth]{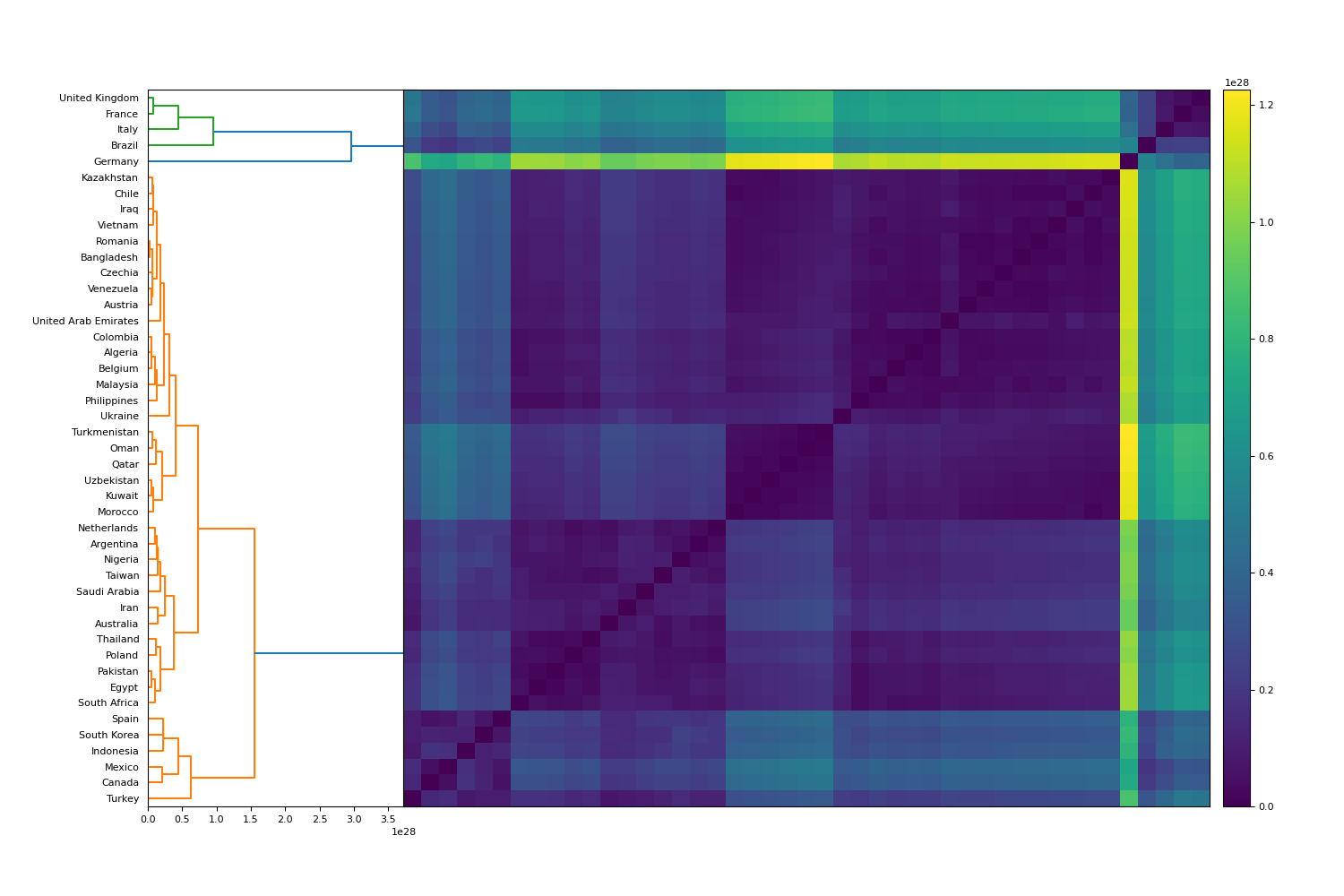}
    \caption{Hierarchical clustering applied to our aggregate normalised distance matrix $\Omega$, which captures similarity in historical trends between countries' real and carbon economies.}
    \label{fig:Clustering_real_carbon_economy}
\end{figure*}

Next, we turn to the aggregated distance matrix, which intends to uncover countries that are similar both in terms of their GDP and CO$_2$ emissions (real and carbon economies, respectively). Multidimensional scaling and the optimal silhouette score identify $k=2$ clusters, revealing two primary classes of countries with respect to their real and carbon economies. Notably, most countries are classified in one cluster, with only a small collection determined to belong to the second cluster. The countries in the minority cluster are Brazil, Canada, France, Germany, Italy, Mexico and the UK. This suggests that the latent factor causing separation in the lower dimensional vector space is related to GDP, rather than population. For instance, Pakistan, which has a population of $\sim$ 220 million people, is determined to exist in the former cluster.

We also present the hierarchical clustering results of our aggregate distance matrix across all three attributes. Figure \ref{fig:Clustering_real_carbon_economy} highlights the existence of one prominent cluster, a small subcluster and an outlier. The primary cluster consists of three subclusters. The first contains countries such as Chile, Iraq, Vietnam, Romania, Bangladesh, and others. Most countries in this cluster are characterised by developing economies, growth in population, and significant growth in CO$_2$ emissions over time. The second subcluster consists of countries such as the Netherlands, Argentina, Nigeria, and Taiwan. These countries mostly display consistent growth in GDP and population and lower levels of CO$_2$ emissions over time. The final subcluster consists of Spain, South Korea, Indonesia, Mexico, Canada and Turkey. Most of these countries displayed moderate growth in GDP and population, and growth in CO$_2$ emissions, exhibiting reasonable variability over time. In particular, many of these countries experience a flattening in emissions over the final 5-10 years of the analysis window. The second cluster includes the UK, France, Italy and Brazil. These countries are characterised by increasing emissions trajectories until the later parts of the analysis window, where a flattening or decline in emissions occurs. Germany is identified as an outlier, which is due to its steady decline in emissions over time and the strong GDP growth over the entire analysis window. It provides a model example of successful economic growth while reducing its environmental impact.

\section{Discussion}
\label{sec:Discussion}

\subsection{Summary of findings}
The central focus of this paper is the evolutionary study of spatio-temporal trends in global and country-specific CO$_2$ emission behaviours over time. In our first section, we demonstrate that most countries are well modelled by a (not necessarily continuous) piecewise linear data generating process. We believe this is the first paper to use such a methodology to study carbon emissions on a country-by-country basis, as well as the first to reveal the extent to which so many countries are well described by this piecewise linear model. Such a finding is not only surprising, it could provoke new research analysing the explanations of such breaks in the rate of growth (or decline) of country emissions. In addition, it would be interesting to monitor moving forward, given the potential for exogenous political and social factors to alter country behaviours. The Russia-Ukraine war of 2022 may cause substantial changes in several countries' emission trends, for example.


Next, we demonstrate pronounced heterogeneity in cluster number, cluster size and constituency among emissions trajectories on a decade-by-decade basis. This section highlights the dynamic nature of this problem and the need for constant monitoring of country behaviours. Then, the evolutionary study of geographic variance reveals a peak in emissions' spatial dispersion in the year 2000, beyond which emissions have become more spatially concentrated. We highlight that this is largely due to China's accounting for a disproportionate level of the global emissions distribution in recent years. We believe this is another original finding revealed by a new methodology and approach. Finally, we introduce a framework to group countries based on their real and carbon economies. Our methodology captures emissions data, GDP and population over the past 50 years, and incorporates dimensionality reduction and clustering. This framework could be applied to other problems, or one could generalise some of the questions explored in this paper using different economic and demographic metrics.

\subsection{Limitations and future work}
There are a variety of avenues researchers could explore when considering avenues for future research. We have applied the concept of geodesic variance to country emissions trajectories. One could extend this analysis to study the geographic spread in emissions among specific industries and identify which sectors of the economy are the most serious violators. Building upon this analysis, one could then study the ``elasticity'' of such emissions and build models to identify which sectors of the economy have emissions that are most easily reduced.

More broadly, given the timeliness and importance of global decarbonisation, a richer dataset with more covariates (such as industry of emissions) could allow policymakers and scientists to build more detailed models for various problems pertaining to prediction and statistical inference. Furthermore, this paper does not consider exogenous factors that may influence emissions; one could explicitly account for societal and political influences in overall emissions behaviours via random or fixed-effects models. Related to this idea, one could build a ranking system of emissions from those countries or industries that provide the most societal or economic benefit per unit omitted. One could conceivably then rank such emissions and determine which are most easily reduced and will have the smallest impact on the economy.

\section*{Competing interests statement}
The authors have no competing interests.

\section*{Data accessibility statement}
All data analysed in this manuscript are available from \cite{CO2data}.

\appendix

\section{Proposition on geographic variance}
\label{app:proposition}

\begin{prop}
Let $f$ be a probability density function on a finite metric space $(X,d)$, and consider its geodesic variance, defined by (\ref{eq:geodesicvariance}). That is,
\begin{align}
\label{eq:geodesicvariance_new}
\text{Var}(f) &= \int_{X \times X} d(x,y)^2 d\mu(x) d\mu(y)\\ &= \sum_{x, y \in X} d(x,y)^2 f(x)f(y).  
\end{align}
When $(X,d)$ is a finite subset of the real number line $\mathbb{R}$ with the Euclidean metric, this quantity reduces to the classical definition of variance, up to a factor of 2.
\end{prop}
\begin{proof}
Let $Y$ be a random variable on the real number line $\mathbb{R}$. The classical definition of variance is the quantity $\text{var}(Y)=\E Y^2 - (\E Y)^2$. If $Y$ has a density $f$ over a (finite) set $X$, this can be expressed as
\begin{align}
    \text{var}(Y) = \sum_{x \in X} x^2 f(x) - \left(\sum_{x \in X} x f(x)\right)^2.
\end{align}
On the other hand, (\ref{eq:geodesicvariance_new}) can be expanded as follows:
\begin{align}
    \text{Var}(f) &= \sum_{x, y \in X} (x-y)^2 f(x)f(y)\\
   &=\sum_{x, y \in X} x^2 f(x)f(y) + \sum_{x, y \in X} y^2 f(x)f(y) - 2\sum_{x, y \in X} xyf(x)f(y)\\
     &= 2\sum_{x, y \in X} x^2 f(x)f(y) - 2\sum_{x, y \in X} xyf(x)f(y)\\
    &=2\sum_{x \in X} x^2 f(x) - 2 \left(\sum_{x \in X} x f(x)\right)\left(\sum_{y \in X} y f(y)\right)\\
    &=2\text{var}(Y).
\end{align}
Thus, up to a factor of 2, the geodesic (or geographic) variance reduces to the classical notion of variance on the real line. 
\end{proof}

With more care, the above proposition holds if $X$ is an arbitrary metric space, with $\mu, \nu$ appropriately integrable measures on $X$.

\bibliographystyle{_elsarticle-num-names}
\bibliography{__newrefs.bib}
\end{document}